\documentclass[sigconf]{acmart}
\usepackage{array,booktabs}
\usepackage{xcolor}
\usepackage[T1]{fontenc}
\usepackage{xspace}
\usepackage{array}
\usepackage{multicol}        
\usepackage{multirow}        
\usepackage{natbib}
\usepackage{algorithm}
\usepackage[]{algpseudocode}
\usepackage{bm}
\usepackage{enumitem}  
\usepackage{mathrsfs}
\usepackage{xurl} 
\usepackage{setspace} 
\usepackage{tikz}
\usetikzlibrary{cd}
\usepackage{balance}

\algrenewcommand\algorithmicrequire{\textbf{Input :}}
\algrenewcommand\algorithmicensure{\textbf{Output :}}

\newtheorem{theorem}{Theorem}
\newtheorem{lemma}[theorem]{Lemma}
\newtheorem{corollary}[theorem]{Corollary}

\newtheorem{proposition}[theorem]{Proposition}
\newtheorem{definition}[theorem]{Definition}

\DeclareMathOperator{\jac}{Jac}

\DeclareMathOperator{\crit}{crit}
\DeclareMathOperator{\sval}{sval}

\DeclareMathOperator{\nprop}{nprop}
\DeclareMathOperator{\spec}{spec}
\DeclareMathOperator{\atyp}{atyp}

\renewcommand{\phi}{\varphi}

\newcommand{\SA}{semi-algebraic\xspace}
\newcommand{\SAC}{connected\xspace}
\newcommand{\SACC}{connected component\xspace}
\newcommand{\SACCs}{connected components\xspace}

\newcommand{\ff}{\bm{f}}
\newcommand{\hh}{\bm{h}}
\newcommand{\bgg}{\bm{g}}
\newcommand{\pp}{{\bm{p}}}
\newcommand{\qq}{\bm{q}}

\newcommand{\bvv}{\bm{v}}

\newcommand{\xx}{{\bm{x}}}
\newcommand{\yy}{{\bm{y}}}
\newcommand{\zz}{{\bm{z}}}

\newcommand{\btheta}{\bm{\theta}}

\newcommand{\V}{\bm{V}}

\newcommand{\CC}{{\mathbb{C}}}
\newcommand{\RR}{{\mathbb{R}}}
\newcommand{\QQ}{{\mathbb{Q}}}

\newcommand{\Ecal}{\mathcal{E}}
\newcommand{\Ecalt}{\widetilde{\mathcal{E}}}
\newcommand{\Gcal}{\mathcal{G}}
\newcommand{\Hcal}{\mathcal{H}}

\newcommand{\Pcal}{\mathcal{P}}
\newcommand{\Qcal}{\mathcal{Q}}

\newcommand{\Vcal}{\mathcal{V}}
\newcommand{\Vcalt}{\widetilde{\mathcal{V}}}

\newcommand{\scrC}{\mathscr{C}}
\newcommand{\scrG}{\mathscr{G}}
\newcommand{\scrGt}{\widetilde{\mathscr{G}}}
\newcommand{\scrP}{\mathscr{P}}
\newcommand{\scrQ}{\mathscr{Q}}

\newcommand{\scrR}{\mathscr{R}}

\newcommand{\zclos}[1]{\overline{#1}^{z}}

\newcommand{\VR}{V_{\mathbb{R}}}

\newcommand{\Map}{\mathcal{R}}
\newcommand{\map}{r}

\newcommand{\maptrig}{x}

\newcommand{\critF}{\crit(\Map, V)}
\newcommand{\svalF}{\sval(\Map, V)}

\newcommand{\propF}{\nprop(\Map, V)}
\newcommand{\atypF}{\atyp(\Map, V)}
\newcommand{\specF}{\spec(\Map, V)}

\newcommand{\FspecF}{\atypF}
\newcommand{\ZFspecF}{\zclos{\atypF}}

\newcommand{\VmCrit}{\VR-\critF}
\newcommand{\VmSpec}{\VR-\specF}

\newcommand{\RmFSpec}{\RR^d-\FspecF}
\newcommand{\RmZFSpec}{\RR^d-\ZFspecF}

\newcommand{\SAS}{\mathscr{S}}

\newcommand{\degV}{\delta}
\newcommand{\degMap}{\mu}
\newcommand{\degtot}{D}
\newcommand{\NQ}{N}

\newcommand{\alert}[1]{\textcolor{red}{#1}}

\newcommand{\et}{\quad \text{and} \quad}
\newcommand{\call}[2]{\textsc{#1}{$\left({#2}\right)$}}

\author{Damien Chablat}
\affiliation{%
	\institution{LS2N, \textsc{CNRS}}
	\city{F-44321 Nantes}
	\postcode{44321}\country{France}}
\email{Damien.Chablat@cnrs.fr}

\author{R\'emi Pr\'ebet}
\affiliation{%
	\institution{Sorbonne Universit\'e, \textsc{CNRS}, \textsc{LIP6}}
	\city{F-75005 Paris}
	\postcode{75252}\country{France}}
\email{remi.prebet@lip6.fr}

\author{Mohab Safey El Din}
\affiliation{%
	\institution{Sorbonne Universit\'e, \textsc{CNRS}, \textsc{LIP6}}
	\city{F-75005 Paris}
	\postcode{75252}\country{France}}
\email{mohab.safey@lip6.fr}

\author{Durgesh H. Salunkhe}
\affiliation{%
	\institution{LS2N, \textsc{CNRS}}
	\city{F-44321 Nantes}
	\postcode{44321}\country{France}}
\email{durgesh.salunkhe@ls2n.fr}

\author{Philippe Wenger}
\affiliation{%
	\institution{LS2N, \textsc{CNRS}}
	\city{F-44321 Nantes}
	\postcode{44321}\country{France}}
\email{Philippe.Wenger@ls2n.fr}

\title{Deciding Cuspidality of Manipulators through Computer Algebra and Algorithms in Real Algebraic Geometry}

\copyrightyear{2022}
\acmYear{2022}
\setcopyright{acmcopyright}\acmConference[ISSAC '22]{Proceedings of the 2022 International Symposium on Symbolic and Algebraic Computation}{July 4--7, 2022}{Villeneuve-d'Ascq, France}
\acmBooktitle{Proceedings of the 2022 International Symposium on Symbolic and Algebraic Computation (ISSAC '22), July 4--7, 2022, Villeneuve-d'Ascq, France}
\acmPrice{15.00}
\acmDOI{10.1145/3476446.3535477}
\acmISBN{978-1-4503-8688-3/22/07}

\begin{document}
\fancyhead{}
\begin{abstract}
  Cuspidal robots are robots with at least two inverse kinematic solutions that can be connected by a singularity-free path. Deciding the cuspidality of generic 3R robots has been studied in the past, but extending the study to six-degree-of-freedom robots can be a challenging problem. Many robots can be modeled as a polynomial map together with a real algebraic set so that the notion of cuspidality can be extended to these data.

In this paper we design an algorithm that, on input a polynomial map in $n$ indeterminates, and $s$ polynomials in the same indeterminates describing a real algebraic set of dimension $d$, decides the cuspidality of the restriction of the map to the real algebraic set under consideration. Moreover, if $D$ and $\tau$ are, respectively the maximum degree and the bound on the bit size of the coefficients of the input polynomials, this algorithm runs in time log-linear in $\tau$ and polynomial in $((s+d)D)^{O(n^2)}$.

It relies on many high-level algorithms in computer algebra which use advanced methods on real algebraic sets and critical loci of polynomial maps. As far as we know, this is the first algorithm that tackles the cuspidality problem from a general point of view.
\end{abstract}
\thanks{%
  The authors are supported by the joint ANR-FWF
  ANR-19-CE48-0015 \textsc{ECARP} project, the ANR grants ANR-18-CE33-0011
  \textsc{Sesame} and ANR-19-CE40-0018 \textsc{De Rerum Natura}
  projects, the European
  Union's Horizon 2020 research and innovation programme under the Marie
  Sk\l{}odowska-Curie grant agreement N. 813211 (POEMA) and the grant FA8665-20-1-7029 of the \textsc{EOARD-AFOSR}}



\maketitle

\section{Introduction}

\paragraph{Problem statement}
Let \(\ff = (f_1, \ldots, f_s)\) be a sequence of polynomials in 
\(\QQ[x_1, \ldots, x_n]\) and $V=\V(\ff)\subset \CC^n$ be the algebraic set it defines (i.e. 
the set of common complex solutions to the $f_i$'s). 
We denote by $\VR=V\cap\RR^n$ the real trace of $V$.
Let $\Map = (\map_1, \ldots, \map_d)$ be a sequence of polynomials in
\(\QQ[x_1, \ldots, x_n]\). 
By a slight abuse of notation, we still denote by $\Map$ the map 
\[
    \Map: \yy\in \CC^n \mapsto (\map_1(\yy), \ldots, \map_d(\yy)) \in \CC^d.
\]
In the whole paper, we make the following assumption:
\begin{enumerate}[label=$(\textsf{\Alph*})$]
    \item\label{ass:A} the ideal generated by \(\ff\), which we denote by
    $\langle\ff\rangle$, is radical and equidimensional of dimension $d$ and
    $\VR$ is not contained in the singular set 
    of $V$.
\end{enumerate}

We denote by $\critF$ the union of the set of \emph{critical points} 
of the restriction of \(\Map\) to \(V\) and the set of \emph{singular points} of $V$ 
(see e.g. \cite[Appendix A.2.]{SS2017} for a definition of these objects).
Further, we denote by $\svalF$ the set of \emph{singular values} of the restriction of
$\Map$ to $V$, i.e. the image by $\Map$ of the set $\critF$: 
\[
    \svalF = \Map(\critF).
\]

Under assumption \ref{ass:A}, the set $\critF$ is the set of common complex 
solutions to the polynomials in \(\ff\) and the set of minors of size $n$ 
of the Jacobian matrix $\jac[\ff,\Map]$ associated to \(\ff, \Map\) (see e.g. \cite[Lemma A.2.]{SS2017}).


The restriction of the map $\Map$ to $V$ is said to be proper at a point $\yy\in\CC^d$ 
if there exists a ball $B\subset\CC^d$ containing $\yy$ such that $\Map^{-1}(B)\cap V$ is 
closed and bounded. The restriction of $\Map$ to $V$ is said to be proper if it is proper
at every point of $\CC^d$.

We denote by $\propF$ be the set of points of $\CC^d$ at which $\Map$ is 
\emph{not} proper. According to \cite[Theorem 3.8.]{Je1999} it is contained in a proper
algebraic set of $\CC^d$.

Finally we denote by $\atypF$ the set of \emph{atypical values} of the restriction of
$\Map$ to $V$, that is the union \(\svalF \cup \propF\), and let 
\[
    \specF=\Map^{-1}(\atypF) \cap V
\]
the set of \emph{special points} of the restriction of $\Map$ to $V$ that map to atypical values.
We denote by $\zclos{\atypF}$ the Zariski closure in $\CC^d$ of the set of atypical values.

Following the formalism introduced in \cite{wenger_new_1992}, we say that the restriction of the map \(\Map\) to $\VR$ is {\em cuspidal} if there exist two distinct points \(\yy\) and \(\yy'\) in \(\VR\) such that the following holds:
\begin{enumerate}[label=(\roman*)]
\item $\Map(\yy) = \Map(\yy')$;
\item there exists a \SACC $C$ of $\VR - \crit(\Map, V)$ which contains both $\yy$ and $\yy'$.
\end{enumerate}
If two such points $\yy$ and $\yy'$ exist, we say that they form a \emph{cuspidal couple} 
of the restriction of $\Map$ to $\VR$. Note that such a couple is not unique in general.

The above definition goes back to some original works in robotics and mechanism design
which we present below. 
The goal of this paper is to design an algorithm which, given as input $\ff$ and $\Map$
as above, decides whether the restriction of $\Map$ to $\VR$ is cuspidal.

\paragraph{Motivations from robotics} 
Cuspidal robots were discovered in the end of the eighties \cite{parenti-castelli_position_1988}. 
A cuspidal robot can move from one of its inverse kinematic solutions to another without meeting a singular configuration. A major consequence is that there is no simple way to know in which solution the robot operates during motion planning trajectories for cuspidal robots is more challenging than for their noncuspidal counterparts \cite{wenger_uniqueness_2004}. Knowing whether a robot under design is cuspidal or not is thus of primary importance. 

Most existing industrial robots are
known to be noncuspidal because they rely on some specific geometric design
rules such as their last three joint axes intersecting at a common point \cite{Wenger97}.
Recently, however, new robots have been proposed that do not follow the
aforementioned design rule, which, in turn, could make them cuspidal (see for e.g.,
\url{https://achille0.medium.com/why-has-no-one-heard-of-cuspidal-robots-fa2fa60ffe9b}).

Hence, obtaining an algorithm for deciding cuspidality is of first importance in this context of 
mechanism design.

\paragraph{Prior works} 
Cuspidal robots have been studied mostly for a specific family of robots
made with three revolute joints mutually orthogonal \cite{wenger_cuspidal_2007}. 
Such robots, were shown to be cuspidal if and only if they have at least one cusp
point in their workspace \cite{el1995recognize, wenger_generic_2022}. Accordingly, an algorithm can be written
that, starting from the inverse kinematic polynomial associated with the robot
at hand, counts the number of triple root of this polynomial. If this number
is nonzero, it means that the robot has at least one cusp and it is thus
cuspidal \cite{corvez_study_2005}. For a general robot, no necessary and sufficient condition 
is known to decide if this robot is cuspidal or not. Thus, no
general algorithm has been devised that can decide if a given arbitrary robot
is cuspidal or not.

The algorithm we design in this paper for deciding cuspidality relies on a family of algorithms 
for solving polynomial systems over the reals with different specifications. Further, we assume 
that all data $\ff$ and $\Map$ have coefficients in $\QQ$ so that bit complexity issues
can be covered without any restriction w.r.t. the application context we target. 

The first routine we use takes as input a polynomial system of $s$ equations and inequalities in 
$\QQ[x_1, \ldots, x_n]$ and returns an encoding of at least one point per connected component of 
the real solution set to the input system. When the input polynomials have degree at most $D$, 
this can be done in time singly exponential in $n$ and polynomial in $D$ and $s$ using the critical point method introduced in \cite{GV1988} and developed in \cite{SS2003,BPR2006,LS2020}. 
The algorithm in \cite{LS2020} is the one which we will specifically use.

The second routine we rely on still takes as input a polynomial system of equations and inequalities, 
as well as the encoding of some query points in the solution set $S \subset \RR^n$ to the input system. 
It then computes an encoding for a semi-algebraic curve, called a roadmap, which has a non-empty 
and connected intersection with all connected components of $S$ and contains all the query points.
This is done in time singly exponential in $n$, polynomial in $D$ and $s$ using more advanced critical 
point methods initiated by Canny in \cite{Ca1988, Ca1988bis, Ca1991, Ca1993} and improved later on in 
\cite{BPR2006,SS2011,BR2014, BRSS2014, SS2017}.

\paragraph{Main results}
In this paper we design an algorithm for deciding the cuspidality on input $\ff$ 
and ${\Map}$ under assumption \ref{ass:A}. Recall that $V = \V(\ff)$ is the algebraic set defined 
by $\ff$ and that $V_\RR = V\cap \RR^n$. When the restriction of the map $\Map$ to $V_\RR$ 
is cuspidal, the algorithm has the ability to output a {\em witness of cuspidality}, i.e. 
a cuspidal couple and an encoding of a semi-algebraic path which connects them in $V_\RR$ 
without meeting $\crit(\Map, V)$. 

Next, we analyze the bit complexity of this algorithm and prove that cuspidality can be decided in time 
singly exponential in $n$, polynomial in the maximum degree of the input polynomials, the integer 
$d$ and $\log$-linear in the maximum bit size of the input coefficients. We use the big-O notation in 
a standard way \cite[Section 3.1]{CLRS2009}. Further, for $\tau\in \RR$, $\tau^\star$ denotes the class
$O(\tau \log(\tau)^a)$ for some constant $a > 0$. 

This leads to the following statement.

\begin{theorem}\label{thm:mainresult}
Let \(\ff = (f_1, \ldots, f_s)\) and $\Map = (\map_1, \ldots, \map_d)$ be two sequences of polynomials in \(\QQ[x_1, \ldots, x_n]\)
, let $V=\V(\ff)$ and $\VR = V\cap\RR^n$.
 Let $\degtot$ be the maximum degree of these polynomials and
 let $\tau$ be a bound on the bit size of the coefficients of the input polynomials.
 Then, under assumption \ref{ass:A}, one can decide the cuspidality of the
 restriction of the map $\Map$ to $\VR$ using at most
 \[
    \tau^\star ((s+d)\degtot)^{O(n^2)}
 \]
 bit operations.
\end{theorem}

We also illustrate how this algorithm runs on classical examples from robotics. 

\paragraph{Structure of the paper}
Section~\ref{sec:prelim} is devoted to recall preliminaries about the subroutines we use and 
Thom's isotopy lemma which is a key ingredient to the correctness proof of our algorithm. Section~\ref{sec:algo} 
is devoted to the formal description of our algorithm and its proof of correctness. The complexity analysis 
is completed in Section~\ref{sec:complexity}. Finally, Section~\ref{sec:example} illustrates how our algorithm
runs on a concrete application from robotics.

\section{Auxiliary algorithms and results}\label{sec:prelim}


\subsection{Sample points algorithms}\label{ssec:sample}
Recall that a semi-algebraic set has only finitely many \SACCs \cite[Theorem 
2.4.4.]{BCR1998}. 
Hence computing at least one point in each of these components constitutes a 
basic subroutine of many algorithms that handle \SA sets.

To encode such points, we use so called {\em zero-di\-men\-sional paramet\-rizations}. 
A zero-dimensional paramet\-rization $\scrP = (\Omega, \lambda)$ is a couple as follows:
\begin{itemize}
\item $\Omega = (\omega, \rho_1, \ldots, \rho_n)$ of polynomials in $\RR[u]$ where $u$ is a 
new variable and $\omega$ is a monic square-free polynomial and $\deg(\rho_i) < \deg(\omega)$;
\item $\lambda$ is a linear form $\lambda_1 x_1 + \cdots + \lambda_n x_n$ in $\RR[x_1, \ldots, x_n]$
\end{itemize}
such that 
\[
\lambda_1 \rho_1 + \cdots + \lambda_n \rho_n = u \frac{\partial \omega}{\partial u} \mod \omega.
\]
Such a data-structure encodes the finite set of points, denoted by $Z(\scrP)$, defined as follows
\[
Z(\scrP) = \left \{\left (\frac{\rho_1(\vartheta)}{\partial \omega / \partial u (\vartheta)}, 
\ldots, 
\frac{\rho_n(\vartheta)}{\partial \omega / \partial u (\vartheta)}
\right) \in\CC^n \mid \omega(\vartheta) = 0\right \}.
\]
We define the {\em degree} of such a parametrization $\scrP$ as the degree of 
the polynomial $\omega$.

We describe a subroutine which takes as input two sequences of polynomials 
$\bgg = (g_1, \ldots, g_s)$ and $\hh = (h_1, \ldots, h_t)$ in $\RR[x_1, \ldots, x_n]$ and outputs 
a sequence of zero-dimensional parametrizations 
\[
    \scrP_1\,, \,\ldots\,,\, \scrP_{r}
\]
such that 
\[
Z(\scrP_1) \cup \cdots \cup Z(\scrP_r) 
\]
has a non-empty intersection with all \SACCs of the semi-algebraic set of $\RR^n$ defined by 
\[
g_1 = \cdots = g_s= 0, \qquad h_1 > 0, \ldots, h_t > 0. 
\]
Further, we denote by $\SAS(\bgg, \hh)\subset \RR^n$ the semi-algebraic set defined by the above systems so that 
$\SAS(0, \hh)$ is the open semi-algebraic set defined by $h_1 > 0, \ldots, h_t > 0$.

We assume that $\bgg$ and $\hh$ have coefficients in $\QQ$ of maximum bit size $\tau$. In that case, the 
polynomials in the output zero-dimensional parametrizations also have coefficients in $\QQ$. We recall the following 
result which allows us to control the cost of computing sample points in semi-algebraic sets. 
\begin{proposition}[{\cite[Algorithm 12.64]{BPR2006}}]\label{prop:sample} 
There exists an algorithm \textsc{SamplePoints} which on input $\bgg$ and $\hh$ as above, with $D$ the maximum degree of 
the $g_i$'s and the $h_i$'s, computes at least one point per \SACCs of $\SAS(\bgg, \hh)$ by means of 
zero-dimensional parametrizations of degree bounded by $D^{O(n)}$ using 
\[\tau (tD)^{O(n)}\]
bit operations.
\end{proposition}
Ideas underlying \textsc{SamplePoints} are the following. First, it considers the hypersurface defined by 
$g = 0$ where $g = g_1^2+ \cdots + g_s ^2$ to handle a unique equation. Next, it introduces 
an infinitesimal $\varepsilon$ to reduce the original problem to the one of computing sample points in each connected 
component of the closed semi-algebraic set defined by 
\[
g = 0, \qquad h_1\geq \varepsilon, \ldots, h_t \geq \varepsilon. 
\]
The latter is done through \cite[Proposition 13.1]{BPR2006} which allows one to 
reduce the original problem to the one of computing sample points in real algebraic sets. 
The latter is done through the so-called {\em critical point method} which consists in computing 
the critical points of a well-chosen polynomial map reaching its extrema on all \SACCs of the considered real algebraic set. 

Such a solving scheme has been refined and improved in particular cases such as the one considered in 
\cite[Section 3.]{LS2020}, where the semi-algebraic set is open and where explicit complexity constants in the big-O 
exponent are well controlled. The following result is a simplification of the statement in \cite[Corollary 3]{LS2020}. 

\begin{corollary}[{\cite[Corollary 3]{LS2020}}]\label{cor:samplerational}
There exists an algorithm \textsc{SamplePointsRational} which on input $\hh$ as above, with $D$ the maximum degree of the $h_i$'s, computes a set of points $\Qcal$ in $\QQ^n$ of cardinality at most
$D^{O(n)}$ and such that $\Qcal$ meets every \SACCs of $\RR^n-\V(\hh)$ using
\[
    \tau (tD)^{O(n)}
\]
bit operations.
\end{corollary}

\subsection{Algorithms for connectivity queries}

We also use algorithms which answer connectivity queries on semi-algebraic sets. 
This is done in two steps. First, on input data which encode a semi-algebraic set $S$ 
under consideration and query points $\Pcal$, one computes a {\em semi-algebraic curve} 
containing $\Pcal$ and whose intersection with all connected components of $S$ is non-empty and 
connected. Hence, we have reduced the original connectivity queries to connectivity queries on a 
semi-algebraic curve. To solve the latter, we rely on classical tools of computer algebra such as 
resultants and real root isolation which are used in algorithms such as the ones in 
\cite{Ka2008,SW2005,JC2021, DMR2012, CJL2013} for this purpose. 

A few words about the encoding of such semi-algebraic curves are in order. Note that a semi-algebraic 
curve is the intersection of an algebraic curve with a given semi-algebraic set. Further, as in e.g. \cite[Section 1.2.]{SS2017} (see also references therein), 
we encode an algebraic curve with 
a {\em one-dimensional rational para\-met\-rization} $\scrR = (\Omega, (\lambda, \mu))$ which is a couple as follows: 
\begin{itemize}
\item $\Omega = (\omega, \rho_1, \ldots, \rho_n)$ of polynomials in $\RR[u, v]$ where $u$ and $v$ are  
new variables and $\omega$ is a monic in $u$ and $v$, square-free polynomial and $\deg(\rho_i) < \deg(\omega)$;
\item $(\lambda, \mu)$ is a couple of linear forms 
\[\lambda_1 x_1 + \cdots + \lambda_n x_n\quad \text{ and }\quad 
\mu_1 x_1 + \cdots + \mu_n x_n
\] 
in $\RR[x_1, \ldots, x_n]$, 
\end{itemize}
such that 
\[
\lambda_1 \rho_1 + \cdots + \lambda_n \rho_n = u \frac{\partial \omega}{\partial u} \mod \omega,
\]
and 
\[
\mu_1 \rho_1 + \cdots + \mu_n \rho_n = v \frac{\partial \omega}{\partial u} \mod \omega.
\]
Such a data-structure encodes the algebraic curve $Z(\scrR)$, defined as the Zariski closure of the following constructible set of $\CC^n$
\[
 \left \{\left (\frac{\rho_1(\vartheta, \eta)}{\partial \omega / \partial u (\vartheta, \eta)}, 
\ldots, 
\frac{\rho_n(\vartheta, \eta)}{\partial \omega / \partial u (\vartheta, \eta)}
\right)
\,\middle|\,
\omega(\vartheta, \eta) = 0, \frac{\partial \omega}{\partial u}(\vartheta, \eta)\neq 0\right \}.
\]
We define the {\em degree} of such a parametrization $\scrR$ as the degree of $\omega$ which coincides with 
the degree of $Z(\scrR)$. Note that such 
a parametrization $\scrR$ of degree $\delta$ involves $O(n\delta^2)$ coefficients. 

As above, we consider sequences of polynomials $\bgg = (g_1, \ldots, g_s)$ and $\hh = 
(h_1, \ldots, h_t)$ in $\QQ[x_1, \ldots, x_n]$ and we let $\SAS(\bgg, \hh)$ be 
the semi-algebraic set defined by 
\[
g_1= \cdots = g_s = 0, \quad h_1 > 0, \ldots, h_t > 0.
\]
We also let $\scrP$ be a zero-dimensional parametrization with coefficients
in $\QQ$. 

We consider an algorithm which, on input $\bgg$, $\hh$ and $\scrP$ computes a 
{\em one-dimensional rational parametrization} $\scrR$ with coefficients in $\QQ$ such that:
\begin{itemize}
\item the finite set of points $Z(\scrP)$ is contained in the algebraic curve $Z(\scrR)$;
\item the intersection of the algebraic curve $Z(\scrR)$ with the semi-algebraic set defined by 
\[
h_1 > 0, \ldots, h_t > 0
\]
is contained in $\SAS(\bgg, \hh)$ and has a non-empty and connected intersection with all its connected components.
\end{itemize}
Such an output is called a {\em roadmap} for the couple $\left(\SAS(\bgg, \hh), Z(\scrP)\right)$
since it designs a semi-algebraic curve which captures the connectivity 
of $\SAS(\bgg, \hh)$ as well as the relative position of all the points in 
$\left(\SAS(\bgg, \hh) \cap  Z(\scrP)\right)$. 
Hence connectivity queries on $\SAS(\bgg, \hh)$ are reduced to connectivity queries on the curve defined by the roadmap.

\begin{proposition}[{\cite{BPR2000}}]\label{prop:roadmap}
Let $\bgg$, $\hh$ and $\scrP$ be respectively two polynomial sequences and a zero-dimensional
parametrization as above. Assume the entries of $\bgg$ and $\hh$ 
have degree bounded by $D$ and let $\delta$ be the degree of $\scrP$.
Let $\tau$ be a bound on the bit size of the input coefficients. There exists an algorithm \textsc{Roadmap} which computes
a one dimensional rational parametrization as above using
\[
\tau^\star t^{O(n)} \delta D^{O(n^2)}
\]
bit operations. 
Besides, the degree of the output rational parametrization is polynomial in $t^{n+1}\delta D^{n^2}$. 
\end{proposition}

On input a description of a semi-algebraic curve as above, answering 
connectivity queries on this curve can be done in time which is 
\emph{polynomial in the degree of the input algebraic curve}. This is done by 
running algorithms that compute a piecewise linear curve that is semi-algebraically homeomorphic
to the curve, and which can be considered as a graph. 
Then, deciding connectivity queries on this curve is reduced to deciding connectivity queries 
on a graph, which is a classically solved algorithmic problem
(see for e.g. \cite[Section 22.2]{CLRS2009}).

An isotopy of $\RR^n$ is an application $\Hcal\colon \RR^n \times [0,1] \to \RR^n$ such that
$\yy\in \RR^n\mapsto \Hcal(\yy,0)$ is the identity map of $\RR^n$ and 
for all $t \in [0,1]$, the map $\yy\in \RR^n\mapsto \Hcal(\yy,t)$ is a homeomorphism.
Then we say that two subsets $Y$ and $Z$ of $\RR^n$ are isotopy equivalent
if there exists an isotopy $\Hcal$ of $\RR^n$ such that $\Hcal(Y,1)=Z$.

\begin{proposition}[{\cite{DMR2012,CJL2013,JC2021}}]\label{prop:isotop}
Let $\scrR$ be a one-dimensional rational para\-met\-rization, $\hh$ a finite sequence of polynomials 
and $\scrP$ a zero-dimensional parametrization such that $Z(\scrP)\subset Z(\scrR)$, all of them with coefficients in $\QQ$. 
Let $\delta_{\scrP}$ and $\delta_{\scrR}$ be the respective degrees of $\scrP$ and $\scrR$ and 
$D$ be the maximum of $\delta_{\scrR}$ and the degrees of the polynomials in $\hh$. 
Let $\tau$ be a bound on the bit size of the coefficients on the input polynomials.

There exists an algorithm \textsc{GraphIsotop} which, on input $\scrR, \hh$ and $\scrP$ 
computes a graph $\scrG = (\Vcal, \Ecal)$, with $\Vcal\subset\RR^n$ such that:
\begin{itemize}
\item the piecewise linear curve $\scrC_\scrG$ associated to $\scrG$, is isotopy equivalent to $Z(\scrR)\cap\SAS(0,\hh)$;
\item the points of $\Vcal$ and $Z(\scrP)\cap\SAS(0,\hh)$ are in one-to-one correspondence 
through the isotopy.
\end{itemize}
Moreover the algorithm outputs a procedure $\textsc{Vert}_{\scrG}$, that on input a
zero-dimensional parametrization $\scrQ$ such that $Z(\scrQ) \subset Z(\scrP)$, computes, using
a number of bit operations polynomial in $\tau\delta_{\scrP}$, 
the subset $\Vcal_{\scrQ}$ of vertices of $\Vcal$ that are associated to
\[
    Z(\scrQ) \cap \SAS(0,\hh).
\]
This is done using at most $\tau^\star(\delta_{\scrP}D)^{O(1)}$ bit operations.
\end{proposition}

Hence, given a graph $\scrG = (\Vcal,\Ecal)$ computed by \textsc{GraphIsotop} the following characterization occurs:
two points of $Z(\scrP)\cap\SAS(0,\hh)$ are connected in $Z(\scrR)\cap\SAS(0,\hh)$ if and only if
the vertices in $\Vcal$, associated to these points, are connected in $\scrG$. 




\subsection{On Thom's isotopy lemma}
In \SA geometry, we are interested about describing and classifying 
the topology of slices of the studied varieties. This is done through 
homeomorphisms we call trivializations.
Let $X$, $Y$ and $Y'$ be \SA sets such that $Y' \subset Y$, and
let $\phi\colon X \to Y$ be a continuous \SA map.
A \SA \emph{trivialization} of $\phi$ over $Y'$ with fiber $F$ is a \SA homeomorphism
$\Psi^{-1}\colon Y'\times F \to \phi^{-1}(Y')$ such that the following diagrams commutes
\[
\begin{tikzcd}
    Y'\times F \arrow[rd, "\pi" below] \arrow[r, "\Psi^{-1}"] & \phi^{-1}(Y') \arrow[d, "\phi"]\\
    & Y'
\end{tikzcd}
\]
where $\pi$ is the projection onto $Y'$. 
We say that $\Psi^{-1}$ is \emph{compatible} with $X'\subset X$ if there is $F'\subset F$ such that
$\Psi^{-1}(Y'\times F') = X' \cap \phi^{-1}(Y')$.

Thom's first isotopy lemma is a classical result of differential 
geometry that allows to construct diffeomorphisms between submanifolds 
\cite{GWDL2006}. In the context of real algebraic geometry, given \SA data, a \SA version of this
theorem has been obtained in \cite[Theorem 1]{CS1995}.
This is done by replacing integration of some vector fields by trivialization of
some proper submersions using a result previously obtained in \cite[Theorem 2.4]{CS1992}. 
We present hereafter a consequence of \cite[Theorem 1]{CS1995} in the framework of our study
that will be ubiquitous in the correctness proof of our algorithm for deciding cuspidality.
The theorem below and some induced properties of the sets in consideration can be related, as done in \cite{Moroz2010}, to the work of \cite{LR2007} on the discriminant varieties, but for a polynomial map instead of projections. We choose here to prove statements adapted to the situation.
\begin{theorem}\label{thm:thoms}
Let \(\ff = (f_1, \ldots, f_s)\) be a sequence of polynomials in \(\RR[x_1, \ldots, x_n]\) 
and $V\subset \CC^n$ be the algebraic set it defines. Suppose that $\ff$ satisfies assumption
\ref{ass:A} and let \(\Map = (\map_1, \ldots, \map_d) \subset \RR[x_1, \ldots, x_n]\).
Then for any \SACC $C$ of $\RmFSpec$ and for any $\pp\in C$, there exists a \SA trivialization
of the restriction of $\Map$ to $\VR$ over $\RR^d$ which is compatible with $C$.
In other words, there exists a homeomorphism
\[
  \Psi = (\Map, \Psi_0) \colon  \Map^{-1}(C)\cap\VR  \to  C \times (\Map^{-1}(\pp)\cap\VR),
\]
such that for every \SACC $H$ of $\Map^{-1}(C)\cap\VR$, 
\[
 \Psi_0(H)=\Map^{-1}(\pp)\cap H,
\]
which is a singleton.
\end{theorem}
\begin{proof}
 Let $C$ be a \SACC of $\RmFSpec$, it is an open semi-algebraic set, which does 
not meet $\svalF$.
 Since $C$ does not meet $\propF$ as well, the restriction 
 $\tilde{\Map}\colon\Map^{-1}(C)\cap\VR\to C$
 is a surjective proper submersion. 
 Then we apply the semi-algebraic version of Thom's isotopy lemma 
\cite[Theorem 1]{CS1995} as follows.
Let $\pp\in C$, there exists a semi-algebraic homeomorphism
\[
    \Psi = (\Map, \Psi_0) \colon  \Map^{-1}(C)\cap\VR  \to  C \times (\Map^{-1}(\pp)\cap\VR).
\]
such that $\Psi^{-1}$ is a semi-algebraic trivialization 
of the restriction of $\Map$ to $\VR$ over $\RR^d$ which is compatible with $C$.
Besides, by assumption \ref{ass:A} and \cite[Lemma A.2.]{SS2017}, for any $\pp \in C$ the
Jacobian matrix of $(\ff,\Map)$ has full rank at all points $\Map^{-1}(\pp)\cap V$, 
so that these fibers are finite.
Let $H$ be a \SACC of $\Map^{-1}(C)\cap\VR$, since $\Psi_0$ is continuous, then so is 
$\Psi_0(H) \subset \Map^{-1}(\pp)\cap V$, which is then, a singleton.
Besides since $\Psi^{-1}$ is a trivialization compatible with $C$, 
with fiber $\Map^{-1}(\pp)\cap V$, then
for any $\yy\in \Map^{-1}(\pp)\cap H$, $\Psi_0(\yy) = \yy$.
Therefore, since $\Map^{-1}(\pp)\cap H$ is a singleton and intersects the singleton $\Psi_0(H)$, they are equal.
\end{proof}

\section{Algorithm}\label{sec:algo}

\subsection{Algorithm description}
We present hereafter Algorithm~\ref{alg:cuspalgo} which takes as input $\ff$ and $\Map$ as above, 
satisfying \ref{ass:A} and which decides the cuspidality of the restriction of $\Map$ to the real
solution set $V_\RR = V\cap \RR^n$ where $V=\V(\ff)$.

It proceeds by computing a zero-dimensional parametrization $\scrP$ of a set
of points that provides cuspidal couples of the restriction of $\Map$ to $\VR$ whenever such a couple exists.
In other words, if no cuspidal couple can be found among $Z(\scrP)$,
then the restriction of $\Map$ to $\VR$ is not cuspidal.

Hence, to solve our cuspidality problem, it suffices to 
compute a graph which is isotopy equivalent to a roadmap of 
$\VmCrit$
connecting the points of $Z(\scrP)$ that lie
in the same \SACC of $\VmCrit$.

In addition to the high-level procedures presented in the previous section, we use here
some basic subroutines to manipulate rational parametrizations, polynomials and graphs.
In the following, $\scrP_{\emptyset}$ will denote a zero-dimensional parametrization
of $\RR^n$ encoding the empty set, and $()$ will denote the empty sequence.
Besides, given a polynomial sequence $\hh=(h_i)_{1\leq i \leq \iota}$ we will note $\pm\hh=(\pm h_i)_{1\leq i \leq \iota}$.

The procedure \textsc{Union} takes as input two zero-dimensional 
para\-metrizations 
$\scrP$ and $\scrP'$ of degree $\delta_{\scrP}$ and $\delta_{\scrP'}$
and returns a zero-dimensional parametrization of $Z(\scrP) \cup Z(\scrP')$ of degree 
$\delta_{\scrP}+\delta_{\scrP'}$.
See \cite[Lemma J.3.]{SS2017} for a description of this procedure.

The procedures \textsc{Crit} and \textsc{AtypicalValues} take as input
a polynomial map $\Map$ and a finite sequence of polynomials
$\hh$. Assuming that $\hh$ satisfies assumption \ref{ass:A}, these two procedures
output finite sequences of polynomials whose complex zero-sets are
respectively $\crit(\Map,\V(\hh))$ 
and a proper subset of $\CC^d$ containing
$\zclos{\atyp(\Map,\V(\hh))}$. We refer to \cite[Lemma A.2]{SS2017} for a description of
\textsc{Crit}. The latter is obtained using more involved algebraic
elimination routine we describe in Section~\ref{sec:complexity}.

Let $\scrG = (\Vcal, \Ecal)$ be a graph and let $v,v' \in \Vcal$ be two vertices.
We say that $v$ and $v'$ are connected in $\scrG$ if there exists a sequence 
$(v_1,\dotsc,v_m)$ of vertices in $\Vcal$ such that for all $1\leq i < m$,
\[
     v_1=v,\quad v_2=v' \et \{v_i,v_{i+1}\} \in \Ecal.
\]
The procedure \textsc{GraphConnected} takes as input $\scrG = (\Vcal, \Ecal)$ and 
$(v,v')$ and outputs \texttt{True} if and only if $v$ and $v'$ are connected in $\scrG$.
Else it outputs \texttt{False}.
This subroutine is classic among graph problems, and can be done using well-know 
algorithms such as the breadth-first search algorithm \cite[Section 22.2]{CLRS2009}.

\begin{algorithm}[h]
\setstretch{1.2}
 \caption{Cuspidality algorithm}\label{alg:cuspalgo}
 \begin{algorithmic}[1] 
  \Require Two sequences \(\ff = (f_1, \ldots, f_s)\) and $\Map = (\map_1, \ldots, \map_d)$ 
  of polynomials in \(\QQ[x_1, \ldots, x_n]\) that satisfy assumption \ref{ass:A}.
  \Ensure  A decision, \texttt{True} or \texttt{False}, on the cuspidality of the restriction of $\Map$ to $\VR=V\cap\RR^n$ where $V=\V(\ff)$.
  \State\label{step:atyp} $\bgg \gets$\call{AtypicalValues}{\Map, \ff};
  \State\label{step:samplerat} $\Qcal \gets$\call{SamplePointsRational}{\bgg};
  \State\label{step:samplevoid} $\scrP \gets \scrP_{\emptyset}$;
  \For{$\qq=(\qq_1,\dotsc,\qq_d) \in \Qcal$} \label{step:forunion}
    \State\label{step:Pq}$\Map_{\qq} \gets (\map_1-\qq_1,\dotsc,\map_d-\qq_d)$;
    \State\label{step:sampleq} $\scrP_{\qq} \gets$\call{SamplePoints}{(\ff, \Map_{\qq}),()};
    \State\label{step:union} $\scrP \gets$\call{Union}{\scrP,\scrP_{\qq}};
  \EndFor\label{step:endforunion}
  \State\label{step:crit} $\Delta \gets$ \call{Crit}{\Map,\ff};
  \State\label{step:rm} $\scrR \gets$\call{Roadmap}{\ff,\pm\Delta,\scrP};
  \State\label{step:isotop} $\Big(\scrG = (\Vcal,\Ecal), \textsc{Vert}_{\scrG}\Big)
  \gets$\call{GraphIsotop}{\scrR,\pm\Delta,\scrP};
  \For{$\qq \in \Qcal$}\label{step:forconnect} 
    \State\label{step:vert} $\Vcal_{\qq}\gets$\call{Vert\textsubscript{$\scrG$}}{\scrP_{\qq}}; 
    \For{$(\bvv_1,\bvv_2) \in \Vcal_{\qq}^2$}\label{step:ifVq}
        \If{\call{GraphConnected}{(\bvv_1,\bvv_2),\scrG} and $\bvv_1\neq\bvv_2$ }\label{step:graphconnect}
            \State\label{step:returntrue} return \texttt{True};
        \EndIf\label{step:endgraphconnected}
    \EndFor\label{step:endifVq}
  \EndFor\label{step:endforconnect}
  \State\label{step:returnfalse} return \texttt{False}.
 \end{algorithmic}
\end{algorithm}

\subsection{Correctness proof}
The correction of Algorithm~\ref{alg:cuspalgo} is stated by the following 
proposition. 
\begin{proposition}\label{prop:corralgo}
 Let \(\ff = (f_1, \ldots, f_s)\) and $\Map = (\map_1, \ldots, \map_d)$ be two sequences  of polynomials in \(\QQ[x_1, \ldots, x_n]\)
 , let $V=\V(\ff)$ and $\VR=V\cap\RR^n$.
 Then, under assumption \ref{ass:A}, the restriction of the map $\Map$ to $\VR$
 is cuspidal if and only if, with inputs $\ff$ and $\Map$,
 Algorithm~\ref{alg:cuspalgo} outputs \normalfont{\texttt{True}}.
\end{proposition}


The rest of this section is devoted to prove this correctness statement.
We assume by now the assumptions of Proposition~\ref{prop:corralgo} to hold.

Note that fibers of the restriction of $\Map$ to $V$ are generically finite by 
\cite[Theorem 1.25]{Sh1994}, and in particular by \cite[Lemma A.2]{SS2017}, for every 
$\pp \in \CC^d-\FspecF$, the fiber $\Map^{-1}(\pp)\cap V$ is finite.

We start by an elementary lemma establishing that two distinct ``regular'' points of $\Map$
on $\VR$, having the same image through $\Map$, must be separated by $\specF$.
\begin{lemma}\label{lem:diffSACCApp}
Let $\yy$ and $\yy'$ be two distinct points of $\VmSpec$ such that 
$\Map(\yy)=\Map(\yy')$. 
Then $\yy$ and $\yy'$ belong to distinct \SACCs of $\VmSpec$.
\end{lemma}
\begin{proof}
 Let us proceed by contradiction and
 suppose there exists a path $\gamma\colon [0,1] \to 
\VmSpec$ such that $\gamma(0)=\yy$ and $\gamma(1)=\yy'$.
By definition, $\Map(\gamma([0,1])) \subset \RmZFSpec$

Let $C$ be the \SACC of $\RmZFSpec$ that contains $\Map(\gamma([0,1]))$. 
According to Theorem~\ref{thm:thoms}, there exists a homeomorphism
\[
 \begin{array}{cccccr}
  \Psi\colon & \Map^{-1}(C)\cap\VR & \to & C &\times &\Map^{-1}(\Map(\yy))\cap\VR
  \\[0.1em]
  & \zz & \mapsto &\big(\Map(\zz)&, & \Psi_0(\zz)\qquad\big)
 \end{array},
\]
such that the image of any \SACC of $\Map^{-1}(C)\cap\VR$, through $\Psi_0$,
is a singleton.
Since $\gamma([0,1])$ is contained in $\Map^{-1}(C)\cap\VR$, then $\yy$ and 
$\yy'$ belong to the same \SACC of $\Map^{-1}(C)\cap\VR$, so that
$\Psi_0(\yy)=\Psi_0(\yy')$. 
Since $\Map(\yy)=\Map(\yy')$, then $\yy=\yy'$ by 
injectivity of $\Psi$. This contradicts the assumption $\yy \neq \yy'$ and 
proves the Lemma.
\end{proof}

In other words, any potential cuspidal couple must contain points from  
different \SACCs of the complementary of $\specF$ in $V_\RR$.
This leads naturally to the following construction that we call here a \emph{cuspidality graph}.

\begin{definition}\label{def:cuspgraph}
 Let $\Vcal \subset \RR^n$ and $\scrG=(\Vcal,\Ecal)$ be a graph. Then we say 
that $\scrG$ is a \emph{cuspidality graph} of the restriction of $\Map$ to $\VR$
if the following holds.
 \begin{enumerate}[label=$(\roman*)$]
  \item The set $\Vcal$ is contained in $\VmSpec$ and intersects every \SACC 
  of $\VmSpec$.
  \item Let $\bvv,\bvv'\in \Vcal$ be such that $\Map(\bvv)=\Map(\bvv')$. 
  Then $\bvv$ and $\bvv'$ are \SAC in $\VmCrit$ if and only if 
  they are in $\scrG$.
  \item Let $\bvv \in \Vcal$, then $\Map^{-1}\left(\Map(\bvv)\right) \cap \VR
  \, \subset \, \Vcal$.
 \end{enumerate}
\end{definition}

Remark that it is straightforward that such a graph exists, 
and, under assumption \ref{ass:A}, it can be supposed to be finite since
$\VmSpec$ has finitely many \SACCs and $\Map$ has finite fibers 
out of $\FspecF$.

Then the following result reduces the problem of deciding the cuspidality of the
restriction of $\Map$ to $\VR$ to a connectivity problem on a finite graph.

\begin{lemma}\label{lem:cuspchara}
 Let $\scrG=(\Vcal,\Ecal)$ be a cuspidality graph of the restriction of $\Map$ to $\VR$.
 Then the restriction of $\Map$ to $\VR$ is cuspidal if and only if there exist 
two distinct vertices $\bvv,\bvv'\in \Vcal$, connected in $\scrG$, and
such that $\Map(\bvv) = \Map(\bvv')$.
\end{lemma}
\begin{proof}
 If such points $\bvv$ and $\bvv'$ exist, they form a cuspidal couple of the restriction of $\Map$
 to $\VR$, so that this map is cuspidal.
 
 Conversely, suppose that the restriction of $\Map$ to $\VR$ is cuspidal so that 
there exist two distinct points $\yy$ and $\yy'$ in $\VmSpec$ having the 
same image through $\Map$ and that belong to the same \SACC $C$ of $\VmCrit$. 
Then, by Lemma~\ref{lem:diffSACCApp}, there exist two distinct \SACCs $H$ and 
$H'$ of $\VmSpec$ such that $\yy\in H$ and $\yy'\in H'$. Remark that both 
$H$ and $H'$ are contained in $C$ since $H$ and $H'$ are two \SAC subsets of 
$\VmCrit$ that have a non-empty intersection with $C$.

By the first item of Definition~\ref{def:cuspgraph}, $\Vcal\cap H$ is not 
empty. 
Then let $\bvv \in \Vcal\cap H$, one has $\bvv \in C$ by the above remark.
Hence, by the second item of Definition~\ref{def:cuspgraph}, one only need to 
prove the existence of $\bvv' \in \Vcal\cap H'$ such that $\Map(\bvv) = \Map(\bvv')$.

Since $H$ is \SAC, there exists a path $\gamma\colon [0,1] \to H$ such that 
$\gamma(0)=\yy$ and $\gamma(1)=\bvv$.
Recalling that $H\subset \VmSpec$, then 
\[
  \Map(\gamma([0,1]) \cap \FspecF = \emptyset.
\]
Let $T$ be the \SACC of $\RmZFSpec$ that contains $\Map(\gamma([0,1]))$. 
According to Theorem~\ref{thm:thoms}, there exists a homeomorphism
\[
 \begin{array}{cccccr}
  \Psi\colon & \Map^{-1}(T)\cap\VR & \to & T &\times &\Map^{-1}(\Map(\yy))\cap\VR
  \\[0.1em]
  & \zz & \mapsto &\big(\Map(\zz)&, & \Psi_0(\zz)\qquad\big)
 \end{array},
\]
such that the image of any \SACC of $\Map^{-1}(T)\cap\VR$, through $\Psi_0$,
is a singleton.
In particular since $\bvv \in H$, then $\Psi(\bvv)=(\Map(\bvv), \Psi_0(\yy))$.

Let $\bvv' = \Psi^{-1}(\Map(\bvv), \Psi_0(\yy'))$. 
By definition, $\Map(\bvv')=\Map(\bvv)$, so 
that by the last item of Definition~\ref{def:cuspgraph}, $\bvv' \in \Vcal$.
Finally, remark that the path
\[
 \begin{array}{cccc}
  \gamma' \colon &[0,1] & \to & \Map^{-1}(T)\cap\VR\\
   &t & \mapsto & \Psi^{-1}(\Map(\gamma(t)), \Psi_0(\yy'))
 \end{array},
\]
is defined for all $t\in[0,1]$ and $\gamma'(0)=\yy' \in H'$. Hence $\bvv' = 
\gamma'(1) \in H'$ since $H'$ is \SAC.

In conclusion, there exist $\bvv$ and $\bvv'$ in $\Vcal$ having the same image 
through $\Map$, such that $\bvv\neq\bvv'$ since $H\cap H'=\emptyset$. Moreover, since $H\cup 
H' \subset C$, then by the second point of Definition~\ref{def:cuspgraph},
$\bvv$ and $\bvv'$ are connected in $\scrG$.
The equivalence is established.
\end{proof}
\begin{figure}[h]
 \includegraphics[width=\linewidth]{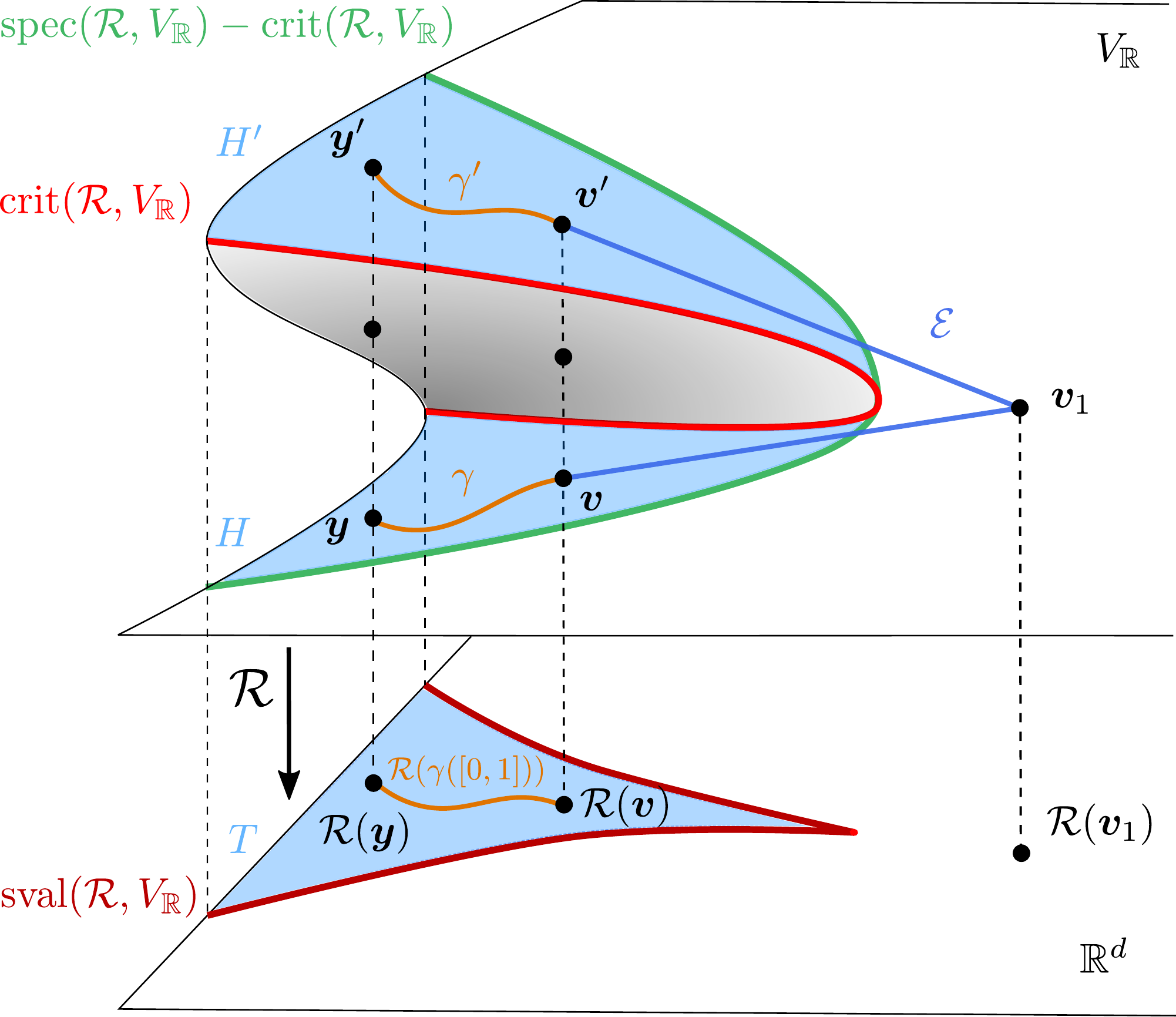}
 \caption{Illustration with $n=3$ and $d=2$ of the proof of Lemma~\ref{lem:cuspchara} where $\Map$ is the projection of the surface $\VR\subset \RR^n$ drawn above the plane $\RR^d$ on the figure. 
 Given a cuspidality graph $\Gcal=\big(\{\bvv,\bvv',\bvv_1\},~\Ecal\big)$ and a cuspidal couple formed by $\yy$ and $\yy'$, one finds, using Theorem~\ref{thm:thoms}, two vertices $\bvv$ and $\bvv'$ that satisfy the statement.}
\end{figure}

Finally, we prove that taking the inverse image of a specific sample set 
of points is enough to satisfy the first item of Definition~\ref{def:cuspgraph}.
\begin{lemma}\label{lem:imagesample}
Let $\Qcal \subset \RR^d$ that intersects every \SACC of $\RmZFSpec$ and let
$\Pcal = \VR \cap \Map^{-1}(\Qcal)$. 
Then $\Pcal$ intersects every \SACC of $\VmSpec$.
\end{lemma}
\begin{proof}
Let $H$ be a \SACC of $\VmSpec$ we need to prove that $H \cap \Pcal$ is not 
empty. Let $\yy \in H$, and let $T$ be the \SACC of $\RmZFSpec$ that contains 
$\Map(\yy)$.
By assumption, there exists $\pp \in \Pcal \cap T$.
Let $\sigma\colon[0,1]\to H$ be a path such that $\sigma(0)=\Map(\yy)$ and 
$\sigma(1)=\pp$. Since $\sigma$ lie in $\RmZFSpec$, the path $\sigma([0,1])$ 
is still contained in $T$.
Then according to Theorem~\ref{thm:thoms}, there exists a homeomorphism
\[
 \begin{array}{cccccr}
  \Psi\colon & \Map^{-1}(T)\cap\VR & \to & T &\times &\Map^{-1}(\Map(\yy))\cap\VR
  \\[0.1em]
  & \zz & \mapsto &\big(\Map(\zz)&, & \Psi_0(\zz)\qquad\big)
 \end{array},
\]
such that the image of any \SACC of $\Map^{-1}(T)\cap\VR$, through $\Psi_0$,
is a singleton.

Let $\gamma \colon t\in [0,1] \mapsto \Psi^{-1}(\sigma(t),\Psi_0(\yy))$, it 
satisfies $\gamma(0)=\yy \in H$.
Since $H$ is \SAC, then $\bvv=\gamma(1)$ belongs to $H$. 
Moreover, since $\sigma(1)=\pp$, then by uniqueness $\Map(\bvv)=\pp$ so that $\bvv \in \Pcal$ and $H\cap \Pcal$ is not empty as claimed.
\end{proof}

We can now proceed to prove the correction of Algorithm~\ref{alg:cuspalgo}.

\begin{proof}[Proof of Proposition~\ref{prop:corralgo}]
Let $\bgg,\Qcal,\scrP, \Delta, \scrR$ and $\scrG=(\Vcal, \Ecal)$ be the data 
obtained in the execution of Algorithm~\ref{alg:cuspalgo}. 
Let us prove that we can derive from $\scrG$ a graph $\scrGt$ that is a
cuspidal graph of the restriction of $\Map$ to $\VR$. 
Then, using this fact and Lemma~\ref{lem:cuspchara},
we prove that the tests on $\scrG$ that are operated in
Algorithm~\ref{alg:cuspalgo}, are enough to conclude on the cuspidality
of the restriction of $\Map$ to $\VR$.
Remark that according to the description of the subroutines
\textsc{AtypicalValues} and \textsc{Crit}, the following holds
\begin{align*}
    &\ZFspecF = \V(\bgg), \quad V\cap\Map^{-1}(\Qcal) =\bigcup_{\qq\in\Qcal}\V(\ff,\Map-\qq)\\
    &\quad\qquad\et \critF = \V(\ff,\Delta).
\end{align*}

Then, according to the first item of Proposition~\ref{prop:isotop} there exists an isotopy 
$\Hcal$ of $\RR^n$ such that $\Hcal(\scrC_{\scrG},1)=Z(\scrR)\cap\RR^n-\critF$ 
where $\scrC_{\scrG}$ is the piecewise linear curve of $\RR^n$
associated to $\scrG$. We denote further $\yy \mapsto \Hcal(\yy,1)$ by $\Hcal_1$.
Let $\Vcalt = \Hcal_1(\Vcal)$ and
\[
    \Ecalt = \left\{ \left\{\Hcal_1(\bvv), \Hcal_1(\bvv')\right\}
    \,\middle|\, \{\bvv,\bvv'\} \in \Ecal \right\}.
\]
Let $\scrGt = (\Vcalt,\Ecalt)$ be the graph thus defined.
According to the second item of Proposition~\ref{prop:isotop} the equality
$\Vcalt=Z(\scrP)\cap\RR^n$ holds since 
\[
    Z(\scrP) \subset \Map^{-1}(\Qcal) \et \Qcal\cap\ZFspecF = \emptyset.
\]
Moreover the following map is a bijection
\[
 \begin{array}{cccccr}
  \Hcal_1\times\Hcal_1\colon&\Ecal & \to & \Ecalt \\[0.1em]
  &\{\bvv,\bvv'\} & \mapsto & \{\Hcal_1(\bvv), \Hcal_1(\bvv')\}
 \end{array}.
\]
Let us show that $\scrGt$ is a cuspidality graph of the restriction of $\Map$ to $\VR$.

By Corollary~\ref{cor:samplerational}, the finite set $\Qcal \subset\RR^d$ intersects
every \SACC of $\RmZFSpec$. Then by Lemma~\ref{lem:imagesample}, every \SACC of
$\VmCrit$ has a non-empty intersection with $\VR\cap\Map^{-1}(\Qcal)$.
As $\Qcal$ is finite and does not intersect $\svalF$, the set $\VR\cap\Map^{-1}(\Qcal)$ 
is a finite union of the sets $\VR\cap\Map^{-1}(\qq)$, which are finite by \cite[Lemma A.2]{SS2017}. 
Hence $\VR\cap\Map^{-1}(\Qcal)$ is finite so that its \SACCs are reduced to its points.
Hence by Proposition~\ref{prop:sample}, $\VR\cap\Map^{-1}(\Qcal)$ is equal to 
$Z(\scrP)\cap\RR^n$ which is itself equal to $\Vcalt$.
Therefore, $\scrGt$ satisfies the first item of Definition~\ref{def:cuspgraph}.

Let $\bvv,\bvv' \in \Vcalt$. According to Proposition~\ref{prop:roadmap}, since
$\bvv$ and $\bvv'$ are in  $Z(\scrP)\cap\RR^n$, they are connected in
$\VmCrit$ if and only if they are connected in 
\[
Z(\scrR)\cap\RR^n-\critF.
\]
However by Proposition~\ref{prop:isotop}, since $Z(\scrP) \subset Z(\scrR)$, then
$\bvv$ and $\bvv'$ are connected in $Z(\scrR)\cap\RR^n-\critF$ if and only if 
$\Hcal_1^{-1}(\bvv)$ and $\Hcal_1^{-1}(\bvv')$ are connected in $\scrG$.
But the latter statement is equivalent to saying that
$\bvv$ and $\bvv'$ are connected in $\scrGt$ since $\Hcal_1\times\Hcal_1$ is a bijection.
Therefore, $\scrGt$ satisfies the second item of Definition~\ref{def:cuspgraph}.

Finally $\scrGt$ satisfies the last item of Definition~\ref{def:cuspgraph} since
for all $\bvv \in \Vcalt$,
\[
    \VR \cap \Map^{-1}(\Map(\bvv)) \:\subset\: \VR \cap \Map^{-1}(\Qcal) = Z(\scrP) \cap \RR^n = \Vcalt.
\]
In conclusion, $\scrGt$ is a cuspidal graph of the restriction of $\Map$ to $\VR$.
Let us prove now that, the restriction of $\Map$ to $\VR$ is cuspidal if and only if, on inputs
$\ff$ and $\Map$, Algorithm~\ref{alg:cuspalgo} outputs \texttt{True}.

If Algorithm~\ref{alg:cuspalgo} outputs \texttt{True}, there exists
$\qq \in \Qcal$ and $\bvv_1,\bvv_2 \in \Vcal_{\qq}$ that are connected in $\scrG$. 
Let $\bvv=\Hcal_1(\bvv_1)$ and $\bvv'=\Hcal_1(\bvv_2)$, then by definition of $\Vcalt$, 
$\bvv$ and $\bvv'$ are in $\Vcalt$.
According to Proposition~\ref{prop:isotop} and the definition of the procedure $\textsc{Vert}_{\scrG}$,
since $\bvv_1,\bvv_2\in \Vcal_{\qq}$, then $\Map(\bvv)=\Map(\bvv')=\qq$.
Besides, by definition of $\Ecalt$, $\bvv$ and $\bvv'$ are connected in $\scrGt$ so that by 
Lemma~\ref{lem:cuspchara}, the restriction of $\Map$ to $\VR$ is cuspidal.

Conversely, suppose that the restriction of $\Map$ to $\VR$ is cuspidal.
Then by Lemma~\ref{lem:cuspchara} there exist two distinct points $\bvv,\bvv'\in\Vcalt$,
connected in $\scrGt$, such that $\Map(\bvv)=\Map(\bvv')$.
Since $\Map(\Vcalt) \subset \Qcal$, there exists $\qq \in \Qcal$ such that 
$\qq=\Map(\bvv)=\Map(\bvv')$.
For such a point $\qq$ let $\scrP_{\qq}$ and $\Vcal_{\qq}$ computed in Algorithm~\ref{alg:cuspalgo}
at respectively step~\ref{step:samplerat} and step~\ref{step:vert}.
Recall that $\scrP_{\qq}$ is the zero-dimensional parametrization encoding $\VR\cap\Map^{-1}(\qq)$
and $\Vcal_{\qq}$ the subset of vertices of $\Vcal$, that are associated to the points
of $\VR\cap\Map^{-1}(\qq)$ through $\Hcal_1$.
Hence according to Proposition~\ref{prop:isotop} and the description of 
$\textsc{Vert}_{\scrG}$, $\Hcal_1^{-1}(\bvv)$ and $\Hcal_1^{-1}(\bvv')$ are
distinct and belong to $\Vcal_{\qq}$.
Since $\bvv$ and $\bvv'$ are connected in $\scrGt$, then so are 
\[
    \Hcal_1^{-1}(\bvv) \et \Hcal_1^{-1}(\bvv')
\]
in $\scrG$. Hence \call{GraphConnected}{(\Hcal_1^{-1}(\bvv),\Hcal_1^{-1}(\bvv')),\Gcal} will
outputs \texttt{True} so that Algorithm~\ref{alg:cuspalgo} outputs \texttt{True}.
\end{proof}

\section{Complexity analysis}\label{sec:complexity}

This section is devoted to the proof of the following
proposition. Together with Proposition~\ref{prop:corralgo}, it
establishes Theorem~\ref{thm:mainresult}.

\begin{proposition}\label{prop:complexity}
Let \(\ff = (f_1, \ldots, f_s)\) and $\Map = (\map_1, \ldots, \map_d)$ be two sequences
of polynomials in \(\QQ[x_1, \ldots, x_n]\) and $\degtot$ be the maximum degree of these polynomials.
Let $\tau$ be a bound on the bit size of the coefficients of the input polynomials.
Then, under assumption \ref{ass:A}, with inputs $\ff$ and $\Map$, the execution of
Algorithm~\ref{alg:cuspalgo} terminates using at most
 \[
    \tau^\star ((s+d)\degtot)^{O(n^2)}
 \]
 bit operations. 
\end{proposition}

\begin{proof}

Fix \(\ff\) and $\Map$, we note $V=\V(\ff)$ and $\VR=V\cap\RR^n$.
Assume that assumption \ref{ass:A} holds that is that
$V$ is equidimensional of dimension $d$.
Let $\degV$ and $\degMap$ be the maximum degree of the polynomials in respectively
$\ff$ and $\Map$ so that $\degtot = \max\{\degV,\degMap\}$, and let $\tau$ be a bound on 
the bitsize of the input coefficients. We proceed by considering each step
of Algorithm~\ref{alg:cuspalgo}.

\paragraph*{Step~\ref{step:atyp}} 
The first step of the algorithm consists in computing polynomials whose complex zero-set is the Zariski
closure of the set of atypical values. According to \cite[Theorem 4.1.]{JK2005}, the set $\FspecF$ is contained in an hypersurface of $\CC^d$ degree bounded by
\[
    \degV^{n-d} \left(n\degV + d(\degMap-\degV) \right)^{d}.
\]
Then the polynomials in the finite sequence $\bgg$, given by the call
to \textsc{AtypicalValues}, have degree bounded by $n^d\degtot^n$. To
compute a polynomial defining them, we rely on the quantifier
elimination algorithm in \cite[Chap. 14]{BPR2006}. Precisely, the set
of non-properness can be defined naturally by a quantified formula
expressing that $y$ is in the set of non-properness if and only if for
any $r>0$
there exists $\epsilon>0$ such that for any $y'\in \RR^d$ and $x'\in
\Map^{-1}(\yy')\cap \VR$, $\|y-y'\|^2< \epsilon$ implies that $\|x'\|>r$. There is one
alternate of quantifiers with blocks of quantified variables of
lengths $1, n+d+1$.
Solving such a quantifier elimination problem is
done using $\tau(s\degtot)^{O((n+d)d)}\subset \tau(s\degtot)^{O(nd)}$
bit operations by \cite[Theorem 14.22]{BPR2006} and it outputs
$(s\degtot)^{O(nd)}$ polynomials of degree in $\degtot^{O(n)}$.
Computing a polynomial encoding the critical values is done still
using quantifier elimination but in an even simpler way: these are the
projections of the values of $\Map$ taken at the system $f_1, \ldots,
f_s$ and the $n-d+1$ minors of the Jacobian matrix associated to $\ff,
\Map$.

\paragraph*{Step~\ref{step:samplerat}}
Since $\ZFspecF=\V(\bgg)$, then by Corollary~\ref{cor:samplerational},
the call to \textsc{SamplePointsRational} outputs a set $\Qcal$
of cardinality $\NQ$ bounded by $n^{O(d^2)}\degtot^{O(nd)}$, using at most 
\[
    \tau n^{O(d^2)}\degtot^{O(nd)}
\]
bit operations. 
We denote further $\Qcal=\{\qq^1,\dotsc,\qq^{\NQ}\}$.

\paragraph*{Steps~\ref{step:forunion}-\ref{step:endforunion}}
Suppose that in the \textbf{for} loop, we consider successively $\qq^1$ to $\qq^{\NQ}$.
Let $0\leq i\leq \NQ$, and let $\delta_{\scrP,i}$ be the degree of $\scrP$ at the
end of the $i$-th iteration.
By Proposition~\ref{prop:sample}, for every $1\leq i\leq \NQ$,
at step~\ref{step:sampleq}, \call{SamplePoints}{(\ff,\Map-\qq_i),0} returns a
zero-dimensional parametrization of degree bounded by $\degtot^{O(n)}$.
Then, we have
\[
    \delta_{\scrP}^i \leq\delta_{\scrP}^{i-1} + \degtot^{O(n)}.
\]
Since $\delta_{\scrP,0} = 0$ then $\delta_{\scrP,\NQ}$ 
is bounded by $n^{O(d^2)}\degtot^{O(nd)}$ since $\NQ$ 
is bounded by $n^{O(d^2)}\degtot^{O(nd)}$.
Since the input has constant size, each call of \textsc{SamplePoints}, at step~\ref{step:sampleq}, costs at most  
$\tau D^{O(n)}$ bit operations.
Besides, since the $\delta_{\scrP,i}$'s are in increasing order, according to \cite[Lemma J.4.]{SS2017}, each call to \textsc{Union}, at 
step~\ref{step:union}, is polynomial in $\delta_{\scrP,\NQ}$.

Therefore, at step~\ref{step:endforunion}, $\scrP$ has degree $\delta_{\scrP}$
bounded by $n^{O(d^2)}\degtot^{O(nd)}$
and the total loop execution is using at most $\tau n^{O(d^2)}\degtot^{O(nd)}$ bit operations.

\paragraph*{Step~\ref{step:crit}}
Next, \call{Crit}{\Map,\ff} returns a sequence of polynomials $\Delta$ by computing
the determinant of all the $n\times n$ submatrices $\jac[\ff,\Map]$ 
according to \cite[Lemma A.2.]{SS2017}.
One sees that there are $\binom{s+d}{n}$ such minors, which have degrees bounded by $n(\degtot-1)$.

\paragraph{Step~\ref{step:rm}}
According to the previous step, and by Proposition~\ref{prop:roadmap},
\textsc{Roadmap}$(\ff,\pm\Delta,\scrP)$ returns a one-dimensional rational parametrization $\scrR$
using at most 
\[
 \tau^\star \binom{s+d}{n}^{O(n)} n^{O(d^2)}\degtot^{O(nd)} (n\degtot)^{O(n^2)}
\]
bit operations which is then bounded by $\tau^\star ((s+d)\degtot)^{O(n^2)}$.
Moreover the degree of $\scrR$ is bounded by 
$((s+d)\degtot)^{O(n^2)}$.

\paragraph*{Step~\ref{step:isotop}}
According to the previous step, and by Proposition~\ref{prop:isotop}, the call to
\textsc{GraphIsotop}, with input $(\ff,\pm\Delta, \scrP)$, costs at most
\[
   \tau^\star ((s+d)\degtot)^{O(n^2)}
\]
bit operations.

\paragraph*{Steps~\ref{step:forconnect}-\ref{step:endforconnect}}
At each iteration, the call to $\textsc{Vert}_{\scrG}$ at step~\ref{step:vert}
requires a number of operations which is polynomial in $\delta_{\scrP}$.
Besides the procedure \textsc{GraphConnected}, who has bit complexity linear in $\delta_{\scrP}$
is called at most $\NQ$ times in the \textbf{for} loop of 
steps~\ref{step:ifVq}-\ref{step:endifVq}.
Hence, the \textbf{for} loop of steps~\ref{step:forconnect}-\ref{step:endforconnect} requires at most $n^{O(d^2)}\degtot^{O(nd)}$
bit operations.

In conclusion the whole execution of Algorithm~\ref{alg:cuspalgo} uses at most
$\tau^\star ((s+d)\degtot)^{O(n^2)}$ bit operations, which proves the proposition.
\end{proof}

\section{An example: Orthogonal 3R serial robot}
\label{sec:example}
The cuspidal behaviour of 3R serial robots has been analyzed extensively in the past \cite{el1995recognize, wenger_changing_1996}. 
In this section, we present an example of an orthogonal 3R serial robot in order to put forth the application of the algorithm. 
Such a robot is modeled as a map that maps the joint angles of the robot to the position of the end-effector. The joint angles belong to the so-called the joint space, while 
the set of the positions of the end-effector is called the workspace. The robot illustrated in this section is similar to the one discussed in \cite{el1995recognize} and is known to be cuspidal. 
Such a robot is defined by its D-H parameters (see \cite{dh_para_ref,wenger_generic_2022}), which are here, following the 
conventions, $d = [0, 1, 0]$, $a = [1, 2, 3/2]$ and
$\alpha = [\pi/2, -\pi/2, 0]$.

\begin{figure}[H]
\centering
\includegraphics[width = 0.23\textwidth]{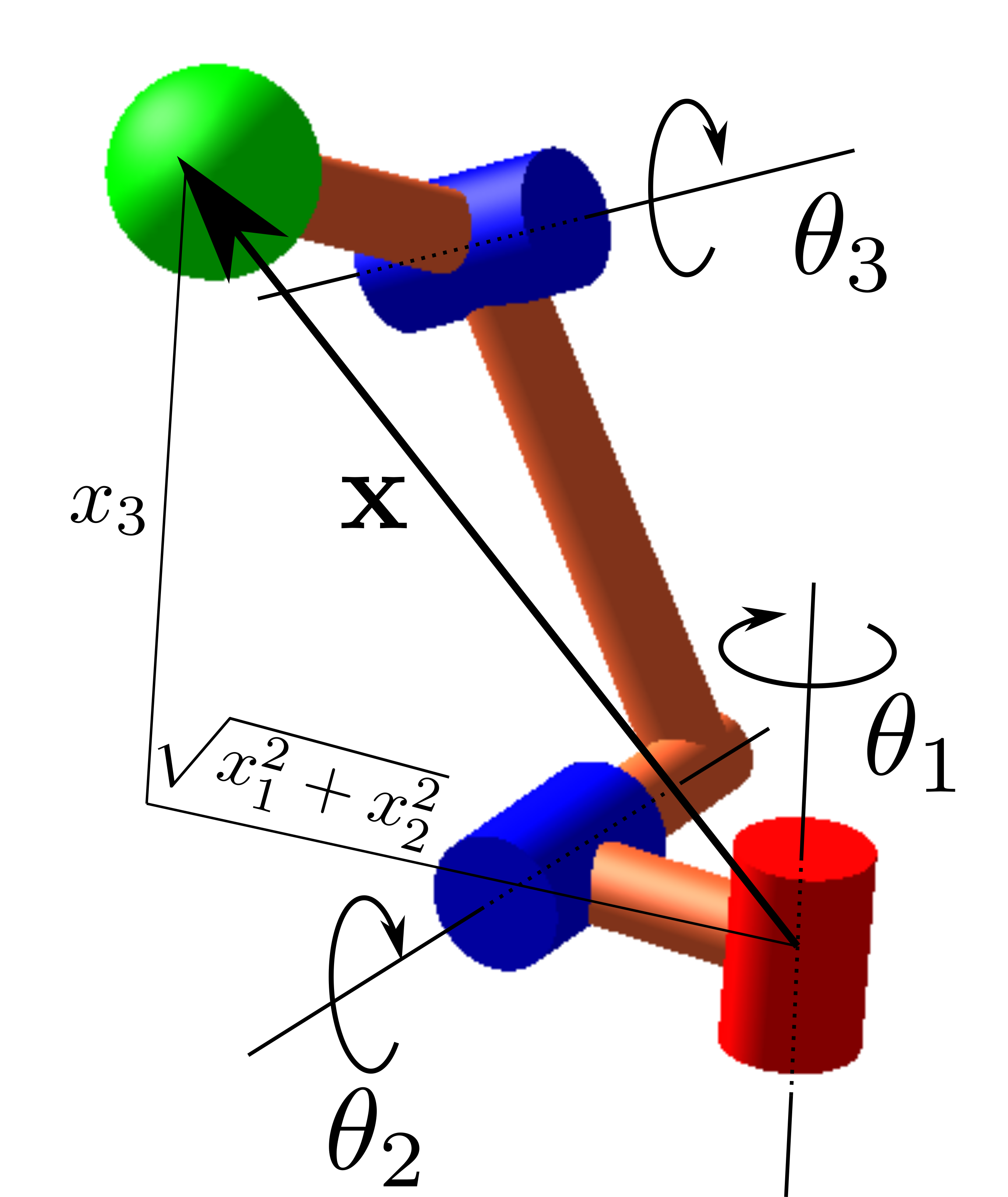}
\caption{An example of orthogonal 3R serial robot}
\label{fig:3R_3d}
\end{figure}
From \cite{wenger_generic_2022}, the robot can be associated to this kinematic map,
\[
\begin{array}{cccc}
    \mathcal{K} \colon& \RR^3 &\longrightarrow & \RR^3 \\
    &\btheta = (\theta_1,\theta_2,\theta_3) &\longmapsto 
    &\big(\maptrig_1(\btheta), \maptrig_2(\btheta),\maptrig_3(\btheta))\big)
\end{array}
\]
where for all $(\theta_1,\theta_2,\theta_3)\in\RR^3$,
\begin{align*}
&\maptrig_1(\theta_1,\theta_2,\theta_3) = \frac{1}{2}c_1c_2(3c_3 + 4) - \frac{1}{2}s_1(3s_3 + 2) + c_1\\
&\maptrig_2(\theta_1,\theta_2,\theta_3) = \frac{1}{2}s_1c_2(3c_3 + 4) + \frac{1}{2}c_1(3s_3 + 2) + s_1\\
&\maptrig_3(\theta_1,\theta_2,\theta_3) = -\frac{1}{2}s_2(3c_3 + 4)
\end{align*}
and for $i\in\{1,2,3\}$, $c_i = \cos(\theta_i)$ and $s_i = \sin(\theta_i)$. 
The singular postures of the robot are the points $(\theta_1,\theta_2,\theta_3) \in\RR^3$
where the determinant of the Jacobian matrix $\jac \mathcal{K}$, of $\mathcal{K}$, vanishes.
Let $\ff = (f_1,f_2,f_3)$ and $\Map = (\map_1,\map_2,\map_3)$  be sequences of polynomials in $\QQ[c_1,s_1,c_2,s_2,c_3,s_3]$ where for all $i\in\{1,2,3\}$ 
\[
 f_i = c_i^2 + s_i^2 -1 \et \map_i = \maptrig_i(\theta_1,\theta_2,\theta_3).
\]
Then, the points $(\theta_1,\theta_2,\theta_3) \in \RR^3$ annihilating $\det(\jac \mathcal{K})$
are exactly the points of $\RR^3$ such that $(c_1,s_1,c_2,s_2,c_3,s_3)\in \VR$ and the matrix 
$\jac[\ff,\Map]$ has not full rank.
Since $\ff$ satisfies assumption \ref{ass:A}, the latter points
are exactly the points of $\crit(\Map,\V(\ff))\cap \RR^n$.

Therefore, the robot can be also modeled as the restriction of the polynomial map associated to $\Map$ to the real algebraic set $\VR = V \cap \RR^n$, where $V=\V(\ff)$, and
deciding the cuspidality of this map amounts to decide the cuspidality of the robot.
Since assumption \ref{ass:A} is satisfied, we can apply Algorithm~\ref{alg:cuspalgo} to $\ff$ and $\Map$ and make this decision.

The set $\critF$ is defined by the vanishing of the following polynomial
\[
 \Delta = -6(3c_3 + 4)(c_2c_3 - 2c_2s_3 - s_3).
\]
Remark that this polynomial does not depend on $c_1$ nor $s_1$.
Since $V$ is bounded by design, the restriction of $\Map$ to $V$ is proper so that $\atypF = \Map(\critF)$. 
Hence the polynomial $g$ whose zero-set is $\zclos{\atyp(\Map)}$ does not depend on $c_1$ nor $s_1$ as well. 
The computation of this polynomial can be done by algebraic elimination and can be found in \cite{el1995recognize}.

The application of Algorithm~\ref{alg:cuspalgo} gives a rise to two main 
sets.
First, the computation of a sample set of points that meets every \SACC of
$\RR^3-\zclos{\atypF}$, is done trough the \textsc{WitnessPoints} function , which is available in Maple 2020. 
The output set $\Pcal$ is represented in Figure~\ref{fig:witness_points} where we adopted a two dimensional representation.
Since $\rho = \sqrt{x_1^2+x_2^2}$ and $x_3$ do not depend on $c_1$ nor $s_1$,
as well as the polynomial $g$ defining $\zclos{\atypF}$, it makes sense to look at the projection
of $\zclos{\atyp(\Map)}$ and $\Pcal$ on the plane associated to $(\rho, x_3)$.

\begin{figure}[H]
\centering
\includegraphics[width = 0.27\textwidth]{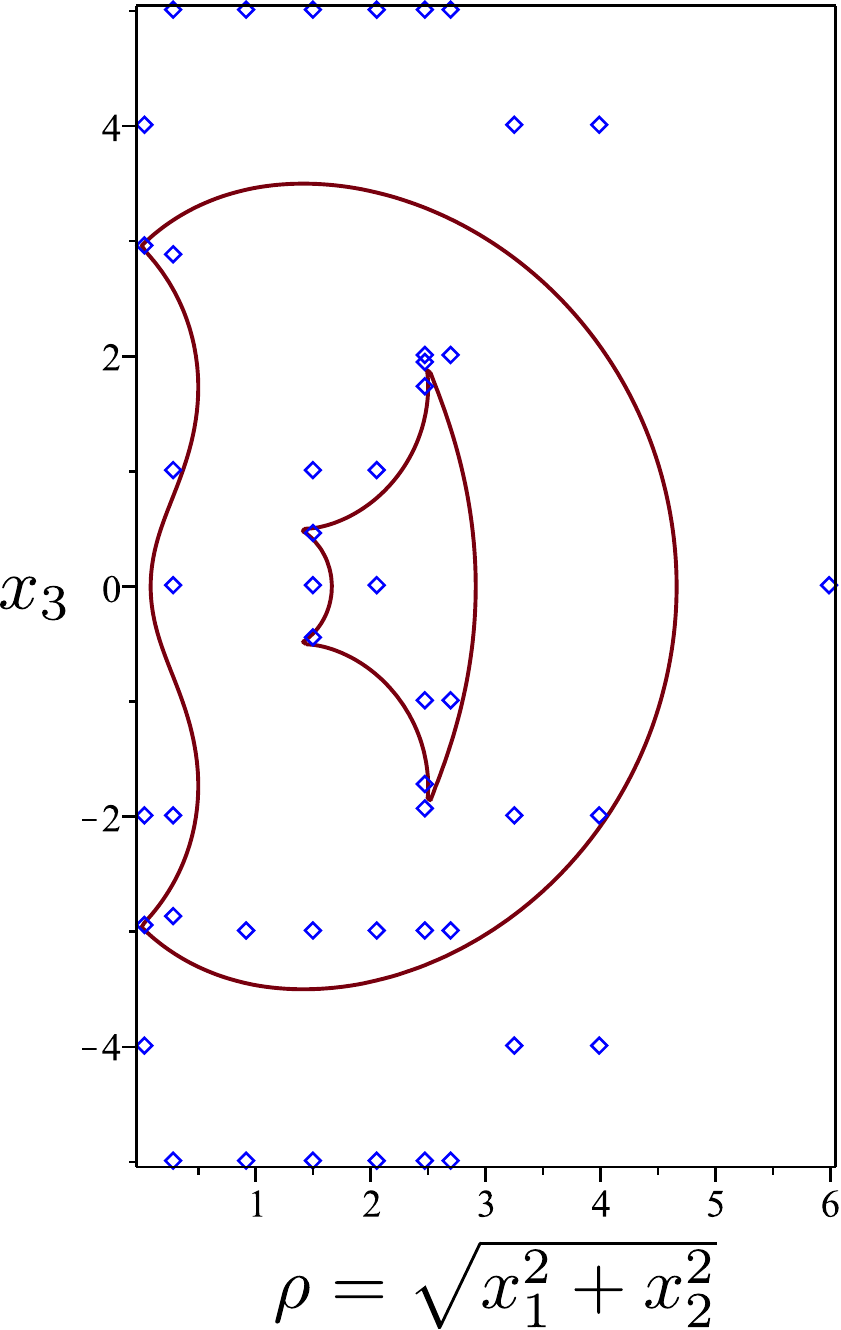}
\caption{Projection on the plane $(\rho, x_3)$ of the set atypical values (red curve) of
an orthogonal 3R serial robot and the points (blue diamonds) of the sample set $\Pcal$
that meets every connected component of the complementary of $\zclos{\atyp(\Map)}$.}
\label{fig:witness_points}
\end{figure}

Then, taking the inverse solutions of these points through $\Map$, we compute a roadmap
of $\VmCrit$ passing through these points. Hence one can easily identify points that belong
to the same \SACC of $\VmCrit$.
Hereafter we describe briefly how do we compute this roadmap. 
The first step consists in deforming the semi-algebraic set $S=\VR-\V(\Delta)$
into the closed semi-algebraic set that is the union of 
\begin{align*}
S^+=\VR\cap\{\xx \in \RR^6 \mid \Delta \geq \epsilon\}\\
\et S^-=\VR\cap\{\xx \in \RR^6 \mid \Delta \leq -\epsilon\}
\end{align*}
with $\epsilon$ small enough.
Since $\VR$ is bounded by design, according to \cite{Ca1993} or 
\cite[Proposition 3.5]{CSS2021}, computing a roadmap of this deformation 
is enough to obtain a roadmap of $S$.
This is done using classical computation of critical loci of projections
and fibers of a projection to repair connectivity failures as described in 
e.g. \cite{Ca1988,Ca1993}.
Moreover we add fibers that pass through the points of $\Pcal$
to determine the \SACC of $S$ where they belong.

In Figure~\ref{fig:connectivity} we draw a roadmap of the projection
of $\VmCrit$ on the plane associated to $(c_2,s_2,c_3,s_3)$ that is
obtained through the above process. Indeed since the polynomial
$\Delta$ does not depend on $c_1$ nor $s_1$, we choose to restrict our
connectivity description on this projection, since extending it to the
whole space is immediate. Finally, since the projection of $\VR$ on
$(c_2,s_2,c_3,s_3)$ is two dimensional, we choose to plot instead the
angles $\theta_1, \theta_2$ that are, modulo $2\pi$, uniquely associated to the data
computed.

We choose here to represent only four inverse solutions of one point
of $\Pcal$ since one can find two cuspidal couples among
them. Indeed, looking at Figure~\ref{fig:connectivity}, one sees that
two dots are on a blue line, while the two others are on a green one.

\begin{figure}[H]
\hspace*{-0.3cm}\includegraphics[width = \linewidth]{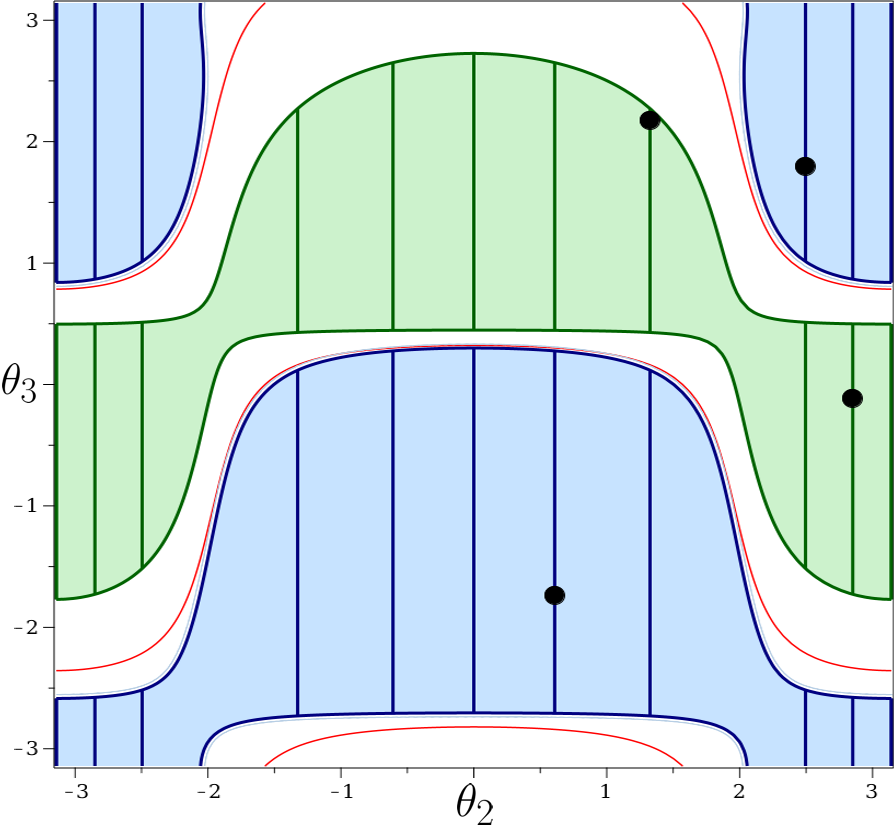}
\caption{The angles that are associated to the projection on the plane associated to $(c_2,s_2,c_3,s_3)$
of the sets under consideration.
The sets $S^+, S^-$ are represented as the areas in respectively green and blue 
while the red line represents the set $\V(\Delta)$.
Besides the black dots are the four inverse images of one sample point of $\Pcal$. 
Finally we represent the roadmap of the projection of $S^+\cup S^-$, containing these points,
as the union of the green and blue lines, which belong to respectively $S^+$ and $S^-$.}
\label{fig:connectivity}
\end{figure}

\bibliographystyle{abbrv}
\balance
\bibliography{biblio.bib}

\end{document}